\documentclass{amsart}
\usepackage{a4wide}
\usepackage{amssymb}
\usepackage{amsmath}
\usepackage{amsfonts}
\usepackage{amsthm}
\usepackage{graphicx}
\usepackage{epsfig}
\usepackage{longtable}
\usepackage[usenames,dvipsnames]{colortbl}
\usepackage[usenames,dvipsnames]{color}
\usepackage{caption}
\usepackage[labelformat=simple,labelfont={}]{subfig}

\newtheorem{thm}{Theorem}%[section]
\newtheorem{lem}[thm]{Lemma}

\newtheorem{prop}[thm]{Proposition}
\newtheorem{conj}{Conjecture}

\theoremstyle{definition}
\newtheorem{defn}[thm]{Definition}

\theoremstyle{remark}
\newtheorem*{rmk}{Remark}

\newcommand{\eps}{\varepsilon}

\newcommand{\DEF}{{:=}}
\newcommand{\FED}{{=:}}

\newcommand{\ii}{\mathrm{i}}

\newcommand{\PT}[1]{\mathbf{#1}}

\newcommand{\re}{\mathop{\mathrm{Re}}}

\DeclareMathOperator{\betafcn}{B}

\DeclareMathOperator{\gammafcn}{\Gamma}

\DeclareMathOperator{\HyperF}{F}

\newcommand{\Hypergeom}[5]{{\sideset{_#1}{_#2}\HyperF\!\left(\substack{\displaystyle#3\\\displaystyle#4};#5\right)}}

\allowdisplaybreaks[1]

\title[A FASCINATING POLYNOMIAL SEQUENCE]
{A FASCINATING POLYNOMIAL SEQUENCE ARISING FROM AN ELECTROSTATICS
PROBLEM ON THE SPHERE}
\author{J. S. Brauchart, P. D. Dragnev\textdagger, E. B. Saff\textdaggerdbl,
and C. E. van de Woestijne\textparagraph} %\textasteriskcentered
\thanks{\noindent The research of this author was supported, in part,
by the Austrian Science Foundation (FWF) under grant S9603-N13. The
author is recipient of an {\sc APART}-fellowship of the Austrian
Academy of Sciences at Vanderbilt (2009) and at UNSW (2010). \\
\textdagger The research of this author was supported, in part, by a Research-in-aid Grant from ORES, IPFW. \\
\textdaggerdbl The research of this author was supported, in part,
by
U.S. National Science Foundation Grant DMS-0808093. \\
\textparagraph The research of this author was
supported by the Austrian Science Foundation FWF under grants S9606 and S9611.}

\date{\today}

\DeclareCaptionLabelFormat{andtable}{#1˜\ #2 and \tablename˜\ \thetable}

\begin{document}

\address{J. S. Brauchart:
School of Mathematics and Statistics,
University of New South Wales,
Sydney, NSW, 2052,
Australia }
\address{E. B. Saff:
Center for Constructive Approximation,
Department of Mathematics,
Vanderbilt University,
Nashville, TN 37240,
USA}
\address{P. D. Dragnev:
Department of Mathematical Sciences,
Indiana-Purdue University,
Fort Wayne, IN 46805,
USA}
\address{C. E. van de Woestijne:
Lehrstuhl f{\"u}r Mathematik und Statistik,
Montanuniversit{\"a}t Leoben,
Franz-Josef-Stra{\ss}e 18,
8700 Leoben, Austria
}
\email{j.brauchart@unsw.edu.au}
\email{dragnevp@ipfw.edu}
\email{Edward.B.Saff@Vanderbilt.Edu}
\email{c.vandewoestijne@unileoben.ac.at}

\begin{abstract}
A positive unit point charge approaching from infinity a perfectly spherical isolated conductor carrying a total charge of $+1$ will eventually cause a negatively charged spherical cap to appear. The determination of the smallest distance $\rho(d)$ ($d$ is the dimension of the unit sphere) from the point charge to the sphere where still all of the sphere is positively charged is known as Gonchar's problem. Using classical potential theory for the harmonic case, we show that $1+\rho(d)$ is equal to the largest positive zero of a certain sequence of monic polynomials of degree $2d-1$ with integer coefficients which we call Gonchar polynomials. Rather surprisingly, $\rho(2)$ is the Golden ratio and $\rho(4)$ the lesser known Plastic number. But Gonchar polynomials have other interesting properties. We discuss their factorizations, investigate their zeros and present some challenging conjectures.
\end{abstract}


\maketitle

\section{Introduction}

Let $\mathbb{S}^{d}$ denote the unit sphere in the Euclidean space
$\mathbb{R}^{d+1}$. Suppose that it is insulated and has a total
positive charge of $+1$. In the absence of an external field the
charge will distribute uniformly with respect to the normalized
surface area measure (unit Lebesgue measure) $\sigma_d$. Now we
introduce a positive unit point charge exterior to the sphere that
repels the charge on the sphere in accordance with the Newton
potential $1/r^{d-1}$, where $r$ represents the distance between
point charges. If this point charge is very close to the sphere,
then one would expect it to cause a negatively charged spherical cap
to appear, while if the point charge is very far from the sphere its
influence is negligible and the charge on the sphere will be
everywhere positive and nearly uniformly distributed over the entire
sphere. We consider the following question: what is the smallest
distance from the unit point charge to $\mathbb{S}^d$ such that the
distribution of the positive charge on the sphere covers all of the
sphere? We will denote this critical distance by $\rho(d)$. As we
shall show, $1+\rho(d)$ equals the largest positive root of the
following polynomial equation of degree $2d-1$:
\begin{equation} \label{polynomial.G}
G(d;z) \DEF \left[ \left( z - 1 \right)^d  - z - 1 \right] z^{d-1} +
\left( z - 1 \right)^d=0, \qquad d = 1, 2, 3, \dots .
\end{equation}
As the question above was communicated to the authors by A. A.
Gonchar, we shall refer to $G(d;z)$ as {\it Gonchar
polynomials}.

Rather surprisingly, $\rho(2)$ turns out to be the {\it Golden
ratio}, which is the limit of the ratio of successive terms in the
{\it Fibonacci sequence} $F_{n}=F_{n-1} +F_{n-2}$, $F_0=F_1=1$ and
$\rho(4)$ is the so-called {\it Plastic number} \cite{St1996}, which
is the limit of ratios of successive terms for the less known {\it
Padovan sequence} $P_n=P_{n-2}+P_{n-3}$, $P_0=P_1=P_2=1$ (sequence
A000931 in Sloane's OEIS~\cite{Sl2003}). That is,
\begin{align*}
\rho(2) &= \lim_{n\to \infty}\frac{F_{n+1}}{F_n}=\frac{1+\sqrt{5}}{2}=1.618033988\dots,\\
\rho(4) &=  \lim_{n\to \infty}\frac{P_{n+1}}{P_n}=\frac{( 9 -
\sqrt{69} )^{1/3} + ( 9 + \sqrt{69} )^{1/3}}{2^{1/3} 3^{2/3}} =
1.3247179572\dots.
\end{align*}
In addition to this curious coincidence, the polynomials $ G(d;z)$
exhibit rather fascinating properties with regard to their
irreducibility over the ring of polynomials with integer
coefficients, as well as the asymptotic behavior of their zeros. Our
goal in the next section is to show how the Gonchar polynomials are
derived and then, in Section 3, to explore some of their properties
and draw the reader's attention to some related conjectures.

\section{Signed Equilibrium}

A general charge distribution on $\mathbb{S}^d=\{\PT{x}\in R^{d+1}\,
:\, |\PT{x}|=1\}
$ will be modeled by a signed measure $\eta$
supported on $\mathbb{S}^d$ with $\eta(\mathbb{S}^d)=1$. The
corresponding {\it Newtonian potential and energy} are given by
\begin{equation}\label{SignedPot}
V^\eta (\PT{x})=\int \frac{1}{|\PT{x}-\PT{y}|^{d-1}}\, d
\eta(\PT{y}), \qquad \mathcal{I}(\eta)=\int V^\eta (\PT{x})\,
d\eta(\PT{x}).\nonumber
\end{equation}
For example, $V^\eta (\PT{x})$ is the familiar Coulomb potential
when $d=2$. For $\eta=\sigma_d$ we expect from the rotational
invariance that the potential $V^{\sigma_d}$ is constant on
$\mathbb{S}^d$ and on concentric spheres. This is a well-known fact in potential theory, but
for completeness we provide a proof below.

\begin{lem} \label{L1}
The potential $V^{\sigma_d}(\PT{x})$ satisfies the following relations
\begin{equation} \label{EquilPotSphere}
V^{\sigma_d}(\PT{x}) = 
\begin{cases}
1 & \text{for $|\PT{x}|\leq 1$,} \\
1 / |\PT{x}|^{d-1} & \text{for $|\PT{x}| > 1$.}
\end{cases}
\end{equation}
\end{lem}
\begin{proof}
Indeed, for any fixed $\PT{x}$ the Newtonian kernel
$f(\PT{y})=|\PT{x}-\PT{y}|^{1-d}$ is a harmonic function
in $\mathbb{R}^{d+1}\setminus \{ \PT{x} \}$,
that is $\Delta f (\PT{y})=\sum_{i=1}^{d+1} f_{y_i y_i}
(\PT{y}) \equiv 0$. Therefore,
by the mean value property for harmonic functions we have %that for $|\PT{x}|>1$,
\begin{equation}
V^{\sigma_d}(\PT{x})=\int_{\mathbb{S}^d} f(\PT{y})\, d
\sigma_d(\PT{y})=f(\PT{0})= \frac{1}{|\PT{x}|^{d-1}} \qquad
\text{{for $|\PT{x}|>1$.}} \label{LebesgueSurfacePotential}
\end{equation}

Considering the sequence of functions $f_n (y)=|(n+1)\PT{x}/n-\PT{y}|^{1-d}$ and applying the monotone convergence theorem allows us to extend this to the unit sphere
giving $V^{\sigma_d}(\PT{x}) \equiv 1$ for $|\PT{x}|=1$. By the
maximum principle for harmonic functions, the potential
$V^{\sigma_d}$ is constant everywhere in the closed unit ball
(``Faraday cage effect''), since it assumes its extreme values 
on the boundary $\mathbb{S}^d$, where it is constant. (One can also verify
this directly by applying the identity $| \PT{x} - \PT{y} |^2 = |
\PT{x} |^2 | \frac{\PT{x}}{| \PT{x} |^2} - \PT{y} |^2$ and
\eqref{LebesgueSurfacePotential}.)
\end{proof}

In the classical Coulomb case ($d=2$), a standard electrostatics
problem (see \cite[Ch.~2]{Ja1998}) is to find the charge
distribution on a charged, insulated, conducting sphere in the
presence of an external field (such as generated by a
positive point charge $q$ off the sphere). This motivates the
following definition.

\begin{defn} Given a continuous function $Q(\PT{x})$ on $\mathbb{S}^d$,
we call a signed measure $\eta_{Q}$ supported on $\mathbb{S}^d$
and of total charge $\eta_{Q}(\mathbb{S}^d)=1$ {\em a signed
equilibrium associated with $Q$} if its weighted Newtonian potential
is constant on $\mathbb{S}^d$; that is, for some constant $F_Q$, 
\begin{equation}
V^{\eta_{Q}}(\PT{x}) + Q(\PT{x}) = F_Q \qquad \text{everywhere on
$\mathbb{S}^d$.} \label{signedeq}\nonumber
\end{equation}
\end{defn}

In the above definition, the function $Q$ is referred to as an {\em
external field}. Thus Lemma \ref{L1} establishes that $\sigma_d$ is
the (signed) equilibrium measure on $\mathbb{S}^d$ in the absence of
an external field. The uniqueness  of the signed equilibrium is
settled by the following known proposition \cite[Lemma~23]{BrDrSa2009}, whose proof is presented in
the Appendix.

\begin{prop} \label{prop:uniqueness}
If a signed equilibrium $\eta_Q$ associated with an external field
$Q$ on $\mathbb{S}^d$ exists, then it is unique.
\end{prop}

Here, we are concerned with the external field generated by a
positive point charge of amount $q$ located at $\PT{a}=R\PT{p}$, $R>1$, where $\PT{p}$ is the North Pole of $\mathbb{S}^d$; that is, 
$\PT{p}=(0,0,\dots,0,1)$. Such a field is given by
\begin{equation} \label{externalfield} 
Q_{\PT{a},q}(\PT{x}) \DEF \frac{q}{|\PT{x}-\PT{a}|^{d-1}}, \qquad \PT{x}\in
\mathbb{R}^{d+1},\end{equation} 
and the associated signed
equilibrium is described in the next result. {(For the general result, see \cite{BrDrSa2009}.)}

\begin{lem}
The signed equilibrium $\eta_{Q}$ associated with the external
field $Q = Q_{\PT{a},q}$ of \eqref{externalfield} is absolutely
continuous with respect to the (normalized) surface area measure on
$\mathbb{S}^d$; that is, $d \eta_Q(\PT{x}) =
\eta_{R,q}^\prime(\PT{x}) d \sigma_d(\PT{x})$, and its density is
given by
\begin{equation} \label{eta.Q}
\eta_{R,q}^\prime(\PT{x}) = 1 + \frac{q}{R^{d-1}} -
\frac{q\left(R^2-1\right)}{\left|\PT{x}-\PT{a}\right|^{d+1}}{,}
\qquad \PT{x}\in \mathbb{S}^d.
\end{equation}
\end{lem}

For the classical Coulomb case this relation is well-known from
elementary physics (cf. \cite[p.~61]{Ja1998}).

\begin{proof} We note first that if $|\PT{x}|>1$, then the function 
$u(\PT{z}) \DEF |\PT{z}-\PT{x}|^{1-d}$ is harmonic in $|\PT{z}|\leq 1$, 
and since the Poisson integral formula \cite{He2009} preserves harmonic functions, we have
\begin{equation} \label{PoissonIdentity} \frac{1}{|\PT{z}-\PT{x}|^{d-1}}=\int_{\mathbb{S}^d} \frac{1}{|\PT{y}-
\PT{x}|^{d-1}}\frac{1-|\PT{z}|^2}{|\PT{z}-\PT{y}|^{d+1}}\,d\sigma_d(\PT{y}), \qquad | \PT{z} | < 1.
\end{equation}
Using a monotone convergence theorem argument as in Lemma \ref{L1}
we extend \eqref{PoissonIdentity} to $|\PT{x}|=1$. We shall make use of the identity 
$|R\PT{y}-\PT{x}|=|\PT{y}-R\PT{x}|$ for any $|\PT{x}|=|\PT{y}|=1$. Then from \eqref{PoissonIdentity} we obtain with $\PT{z}=\PT{p}/R$ and $| \PT{x} | = 1$,
\begin{equation}\label{PoissonComp}
\begin{aligned}
Q( \PT{x} ) = \frac{q}{|R\PT{p}-\PT{x}|^{d-1}}=&\frac{q}{|\PT{p}-R\PT{x}|^{d-1}}=\frac{q}{R^{d-1}|\PT{z}-\PT{x}|^{d-1}}=\int_{\mathbb{S}^d} \frac{q}{R^{d-1}|\PT{y}-
\PT{x}|^{d-1}}\frac{1-|\PT{z}|^2}{|\PT{z}-\PT{y}|^{d+1}}\,d\sigma_d(\PT{y})\\
=&\int_{\mathbb{S}^d} \frac{q}{|\PT{y}-
\PT{x}|^{d-1}}\frac{R^2-1}{|\PT{p}-R\PT{y}|^{d+1}}\,d\sigma_d(\PT{y})=
\int_{\mathbb{S}^d} \frac{1}{|\PT{y}-
\PT{x}|^{d-1}}\frac{q(R^2-1)}{|\PT{a}-\PT{y}|^{d+1}}\,d\sigma_d(\PT{y}).
\end{aligned}
\end{equation}
The Poisson integral formula applied to the constant function $u(\PT{z}) \equiv q/R^{d-1}$ and $\PT{z}=\PT{p}/R$ yields
\begin{equation*}
\frac{q}{R^{d-1}}=\int_{\mathbb{S}^d} \frac{q(R^2-1)}{|\PT{a}-\PT{y}|^{d+1}}\,d\sigma_d(\PT{y}).
\end{equation*}
This together with \eqref{PoissonComp} implies that for $d \nu \DEF \eta_{R,q}^\prime d \sigma_d$ we have
\begin{equation*}
\int_{\mathbb{S}^d}\, d \nu = 1 \qquad \text{and} \qquad V^{\nu}(\PT{x})+Q(\PT{x})=1+\frac{q}{R^{d-1}}, \quad \PT{x} \in \mathbb{S}^d,
\end{equation*}
which proves the lemma.
\end{proof}

\begin{rmk}
A surprising aspect of the equilibrium support, which varies with $R$, is illustrated in Figure~\ref{fig1} for $\mathbb{S}^2$ and $q=1$. 
Writing $\eta_{Q_{\PT{a},1}} = \eta_{Q} = \eta_{Q}^+ - \eta_{Q}^-$ in terms of its Jordan decomposition, it is clear on geometrical and physical grounds that the support of $\eta_Q^+$ is a spherical cap centered at the South Pole, while its complement (a spherical cap centered at the North Pole) is the support of $\eta_Q^-$, see Figure~\ref{fig1a}. Although somewhat counter-intuitive, we see from Figure~\ref{fig1b} that, as the unit charge approaches the sphere ($R \to 1$), from a certain 
distance on, the support of $\eta_Q^+$ occupies increasingly more of
the sphere. This phenomenon can be explained by the geometry of the sphere, where an increasingly 'needle-like'
negative sink forming at the North Pole is balanced by an increasingly
uniform positive part so that the total charge is always one.
\end{rmk}

\begin{figure}[h!t]
\begin{center}
\subfloat[]{\label{fig1a}\includegraphics[scale=.5]{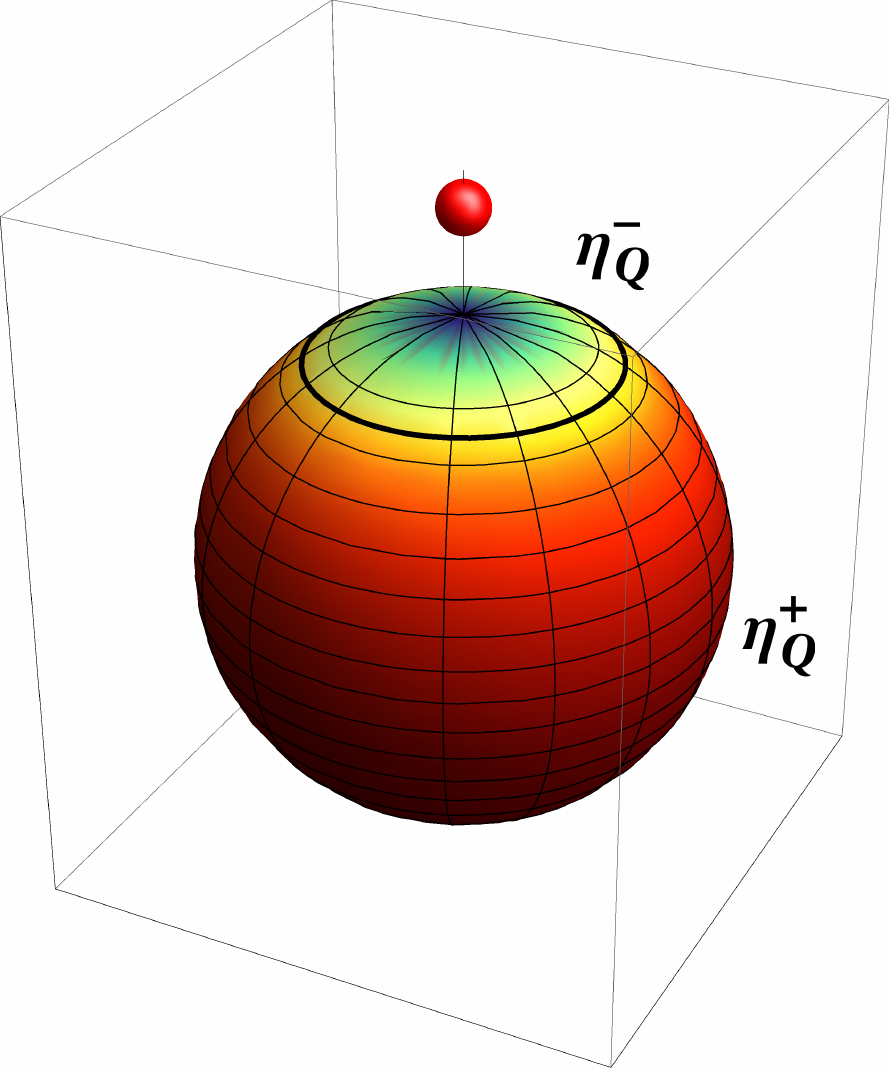}} 
\hfill
\subfloat[]{\label{fig1b}\includegraphics[scale=1]{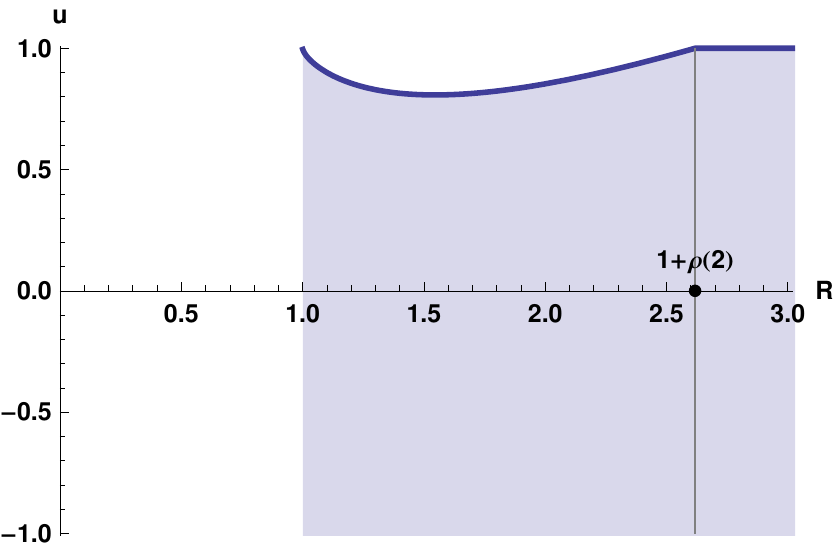}}
\hfill
\caption{\label{fig1} Signed equilibrium $\eta_Q$ on $\mathbb{S}^2$
(left) and projection onto the polar axis of the positive part $\eta_Q^+$.}
\end{center}
\end{figure}

From formula \eqref{eta.Q} we observe that the  minimum value of the density
$\eta_{R,q}^\prime(\PT{x})$ is attained at the North Pole $\PT{p}$,
and its value there is
\begin{equation}\label{eta:eq}
\eta_{R,q}^\prime(\PT{p}) = 1 + \frac{q}{R^{d-1}} -
\frac{q\left(R^2-1\right)}{\left|\PT{a}-\PT{p}\right|^{d+1}} = 1 +
\frac{q}{R^{d-1}} -
\frac{q\left(R^2-1\right)}{\left(R-1\right)^{d+1}}.
\end{equation}

This leads to the following theorem.

\begin{thm} \label{thm:Gonchar}
For the external field $Q_{\PT{a},q} (\PT{x})$ of
\eqref{externalfield} with $\PT{a}=R\PT{p}$, the signed equilibrium
is a positive measure on all of $\mathbb{S}^d$ if and only if $R\geq
R_q$, where $R_q$ is the unique (real) zero in $(1,+\infty)$ of the
polynomial
\begin{equation} \label{equation}
G(d,q;z) \DEF \left[ \left( z - 1 \right)^d / q - z - 1 \right] z^{d-1} + \left( z - 1 \right)^d.
\end{equation}
In particular, the solution to Gonchar's problem is given by
$\rho(d)=R_1 -1$.
\end{thm}

\begin{proof} 
We use the fact that the support of $\eta_{Q_{\PT{a},q}}$ is all of $\mathbb{S}^d$ if and only if $\eta_{R,q}^\prime(\PT{p}) \geq 0$ in \eqref{eta:eq} or equivalently that $G(d,q;R) \geq 0$. Hence we seek the number $R > 1$ such that $G(d,q;R) = 0$. Observing that $G(d,q;1) < 0$ and $G(d,q;x) > 0$ for $x>1$ sufficiently large, there exists at least one value $R_q$ such that $G(d,q; R_q) = 0$. Moreover, this root in $(1, \infty)$ is unique as can be seen by applying Descartes' Rule of Signs to 
\begin{equation*} %\label{eq.expansion.a}
G(d,q; 1 + w) = \frac{1}{q}\sum_{m=0}^{d-1} \binom{d-1}{m} w^{m+d} - \sum_{m=0}^{d-1} \left[ \binom{d}{m} + \binom{d-1}{m} \right] w^m.
\end{equation*}
\end{proof}

Curiously, for $d=2$ and $d=4$ we obtain the Golden
ratio and the Plastic number as answers to Gonchar's problem as mentioned in the Introduction.
Furthermore, for large values of $d$ the asymptotic analysis (provided in the Appendix) shows that
\begin{equation} \label{AsymptoticEquation}
R_q = 2 + \left[ \log (3 q) \right] / d + \mathcal{O}(1/d^2) \qquad
\text{as $d\to\infty$.}
\end{equation}
The appearance of $2$ as the limit as $d\to\infty$ of the
critical distances $R_q$ ($q>0$ fixed) can be explained by studying the asymptotic behavior of
$G(d,q; R)$ as $d \to \infty$. Indeed, for $R \geq 2 + \eps$ ($\eps > 0$
fixed) it goes to $+\infty$ as $d\to\infty$, while for
$1 < R \leq 2 - \eps$ ($\eps > 0$ fixed) it goes
to $-\infty$, leaving the number $2$ as the only
candidate.

\section{The polynomials $G(d;z)$}
\label{sec2}

In the following we investigate the family of Gonchar polynomials $G(d,z) = G(d,1;z)$ given in \eqref{polynomial.G}. Aside from the solution to Gonchar's problem, these polynomials are interesting in themselves and their distinctive properties merit further study.

A polynomial $P$ with real coefficients is called {\em
(self-)reciprocal} if its {\em reciprocal polynomial} $P^*(z) \DEF
z^{\deg P} P(1/z)$ coincides with $P(z)$. In other words, the
coefficients of $z^k$ and of $z^{\deg P - k}$ in $P(z)$ are the
same. Notice that {\em $G(d; z)$ is self-reciprocal} for even $d$ since
\begin{equation*}
z^{2 d - 1} G(d; 1/z) = \left[ \left( 1 - z \right)^d - z - 1
\right] z^{d-1} + \left( 1 - z \right)^d.
\end{equation*}
Consequently, if $\zeta$ is a zero of $G(d; z)$, then so is $1 /
\zeta$ for even $d$. For odd $d$ we infer from
\begin{equation*}
G(d; z) + G^*(d; z) = \left[ 1 + ( -1 )^d \right] \left( z - 1
\right)^d \left( z^{d-1} + 1 \right) - 2 \left( z^d + z^{d-1}
\right) = - 2 \left( z^d + z^{d-1} \right)
\end{equation*}
that the coefficients of $z^k$ and $z^{2 d - 1 - k}$ in $G(d; z)$
sum to zero except for the 'innermost' pair.

\subsection{Factorizations and Irreducibility}
With the aid of symbolic computation programs one can find factorizations and check irreducibility of explicitly given polynomials. For $d = 1, 2, \dots, 7$, we thereby obtain the following factorizations over the integers of $G(d; z)$.
\begin{align*}
G(1; z) &= z-3, \\
G(2; z) &= (z+1) \left(z^2-3 z+1\right), \\
G(3; z) &= z^5-3 z^4+3 z^3-5 z^2+3 z-1, \\
G(4; z) &= (z+1) \left(z^3-3 z^2+2 z-1\right) \left(z^3-2 z^2+3 z-1\right), \\
G(5; z) &= z^9-5 z^8+10 z^7-10 z^6+5 z^5-7 z^4+10 z^3-10 z^2+5 z-1, \\
G(6; z) &= (z+1) \left(z^2-z+1\right) \left(z^8-6 z^7+15 z^6-21 z^5+21 z^4-21 z^3+15 z^2-6 z+1\right), \\
G(7; z) &= z^{13}-7 z^{12}+21 z^{11}-35 z^{10}+35 z^9-21 z^8+7 z^7-9 z^6+21 z^5-35 z^4+35 z^3-21 z^2+7 z-1. \\
\end{align*}

One can easily verify that $(z + 1)$ divides $G(d; z)$ if and only if $d$ is even. Furthermore, the factor $z^2 - z + 1$ arises in the following cases.
\begin{prop} \label{prop:cyclot.divisor}
The cyclotomic polynomial $z^2-z+1$ divides $G(d;z)$ if and only if $6$ divides $d$.
\end{prop}

\begin{proof}
Note that if $\zeta$ is a zero of $z^2-z+1$, then $\zeta^2=\zeta-1$ and $\zeta^3 = -1$. Using formula \eqref{polynomial.G}, it readily follows that $G(d;\zeta)=0$ whenever $d\equiv 0 \pmod{6}$ and $G(d;\zeta) \neq 0$ otherwise.
\end{proof}

A general irreducibility result has so far eluded the authors. The Eisenstein
criterion and, in general, reduction to finite fields, seem not to be effective tools for studying the polynomials $G(d;z)$. However, using Mathematica, we verified the following conjecture for $d$ up to $500$.

\begin{conj}
Set $\ell(d; z) \equiv 1$ for $d$ odd, $\ell(d; z) \DEF z + 1$ for $d$
even but not divisible by $6$, and $\ell(d; z) \DEF ( z + 1 ) (
z^2 - z + 1 )$ if $6$ divides $d$. Then $G(d; z) / \ell(d; z)$ is
irreducible over the rationals except for $d = 4$, $8$ and $12$.
\end{conj}

Regarding the exceptional cases we record that, in addition to $G(4; z)$ as given above,
\begin{align*}
G(8; z) &= (z+1) \left(z^4-3 z^3+3 z^2-3 z+1\right) \\
&\phantom{=}\times \left(z^{10}-6 z^9+16 z^8-24 z^7+24 z^6-21 z^5+24 z^4-24 z^3+16 z^2-6 z+1\right), \\
G(12; z) &= (z+1) \left(z^2-z+1\right) \left(z^6-4 z^5+5 z^4-3 z^3+5 z^2-4 z+1\right) \\
&\phantom{=}\times \big(z^{14}-8 z^{13}+29 z^{12}-62 z^{11}+85 z^{10}-77 z^9+48z^8 \\
&\phantom{=\times}-33 z^7+48 z^6-77 z^5+85 z^4-62 z^3+29 z^2-8
z+1\big).
\end{align*}

\subsection{Zeros of Gonchar polynomials}

Figure~\ref{fig5} illustrates some features of $G(d;z)$ on the
real line and suggests some of its general properties. Note the qualitatively
different behavior for even and odd $d$.

\begin{figure}[h!t]
\begin{center}
\includegraphics[scale=.85]{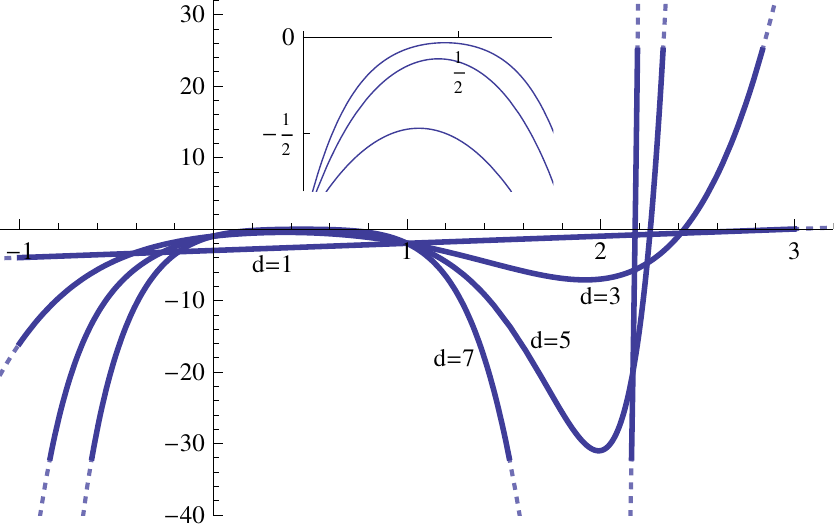}
\includegraphics[scale=.85]{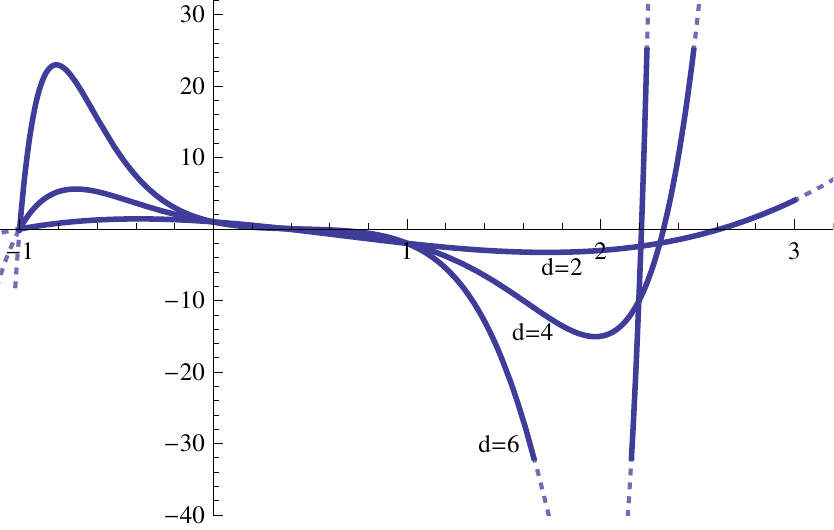}
\caption{\label{fig5} The polynomial $G(d;z)$ on the real line for
$d=1,2,\dots,7$.}
\end{center}
\end{figure}

Depending on the parity of $d$, the polynomial equation $G(d; z) =
0$ has either one ($d$ odd) or three ($d$ even) real simple roots.
More precisely, the following holds.
\begin{prop} \label{prop:real.zeros}
If $d$ is odd, then $G(d; z)$ has precisely one real zero, which is simple and lies in the interval $(2,3]$. If $d$ is even, then $G(d; z)$ has exactly three real zeros: one at $-1$, one in the interval $(1/3, 1/2)$ and one in the interval $(2, 3)$; all these zeros are simple.
\end{prop}

\begin{proof}
Note that for $d \geq 2$ there holds: %\COMMENT{Eq. number del}
\begin{equation*} %\label{eq:special.values}
G(d; -1) = ( -2 )^d \left[ 1 - ( -1 )^d \right], \quad G(d; 0) = (
-1 )^d, \quad G(d; 1) = -2, \quad G(d; 2) = 1 - 2^d.
\end{equation*}
The assertion of the proposition is trivial for $G(1; z) = z - 3$. So let $d \geq 2$. Since $G(d;2) < 0$ and $G(d;3) > 0$, the polynomial $G(d;z)$ has at least one real zero in the interval $(2,3)$ and, as we observed in the proof of Theorem~\ref{thm:Gonchar}, this is its only zero on $[1, \infty)$ and must be simple. For odd $d$, each of the terms $[(x-1)^d-x-1]x^{d-1}$ and $(x-1)^d$ is negative for $x < 1$ and hence so is $G(d; x)$; thus $G(d;x)$ has no zeros outside $(2,3)$. For even $d$, the self-reciprocity of $G(d;x)$ implies that to each zero $\xi(d)$ of $G(d;x)$ in $[1,\infty)$ there is a zero $\xi^*(d) = 1 / \xi(d)$ in $(0,1]$. Consequently, by the first part of the proof, $G(d; x)$ has one and only one zero in $(0,1)$ and this zero is simple and lies in the interval $(1/3, 1/2)$. It remains only to consider the interval $(-\infty, 0]$. Clearly $G(d;-1) = 0$, and by computing $G^\prime(d;x)$ and analyzing its sign (in particular, $G^\prime(d;-1) > 0$), one can show that $G(d;x)$ is strictly increasing on $(-\infty,-1)$. Thus $G(d;x) < 0$ on $(-\infty, -1)$ and, by self-reprocity, $G(d;x) > 0$ on $(-1,0)$.
\end{proof}

Regarding the behavior as $d$ increases of the zeros in $(2,3)$, M. Lamprecht \cite{La2011} proved the following.
\begin{prop}
The zeros $\xi(d)$ of $G(d; z)$ in the interval $(2,3]$ form a
strictly monotonically decreasing sequence with limit point $2$
(compare with \eqref{AsymptoticEquation}).
\end{prop}

We now turn to the study of the complex zeros of $G(d;z)$. Observe that by Proposition~\ref{prop:real.zeros}, $G(d; z)$ has either $d - 1$ (if $d$ is odd) or $d - 2$ (if $d$ is even) pairs of complex conjugated zeros (counting multiplicity). In Figure~\ref{fig6} we have plotted the zeros of $G(d;z)$ for $d = 9, 10, 11,$ and $12$ along with the two unit circles $\mathcal{C}_0$, $\mathcal{C}_1$ centered respectively at $0$ and $1$. Notice that these circles intersect at the points $(1 \pm i \sqrt{3} ) / 2$ and, by Proposition~\ref{prop:cyclot.divisor}, these points are zeros of $G(d; z )$ if and only if $d$ is a multiple of $6$. The zeros of $G(d;z)$ seem to occur roughly into three categories: zeros close to $\mathcal{C}_0$ (indicated by $\times$), zeros close to $\mathcal{C}_1$ (indicated by $+$) and zeros close to the vertical line $x = 1 / 2$ (indicated by $\blacklozenge$). The numbers $N_1$, $N_2$, $N_3$ of zeros in each of the categories are listed in Table~1 for $d = 1, \dots, 12, 42$, from which it appears that these numbers are nearly the same. We will discuss this further in Section~\ref{sec:asymptotics} (see Conjecture~\ref{conj:unit.circle}).

Figure~\ref{fig6} suggests that for even $d$ the polynomial $G(d;z)$ may have zeros lying precisely on $\mathcal{C}_0$. This, in fact, is the case as we now prove.
\begin{figure}[h!t]
\begin{center}
\hfill
\begin{minipage}{0.7\linewidth}
\includegraphics[scale=.625]{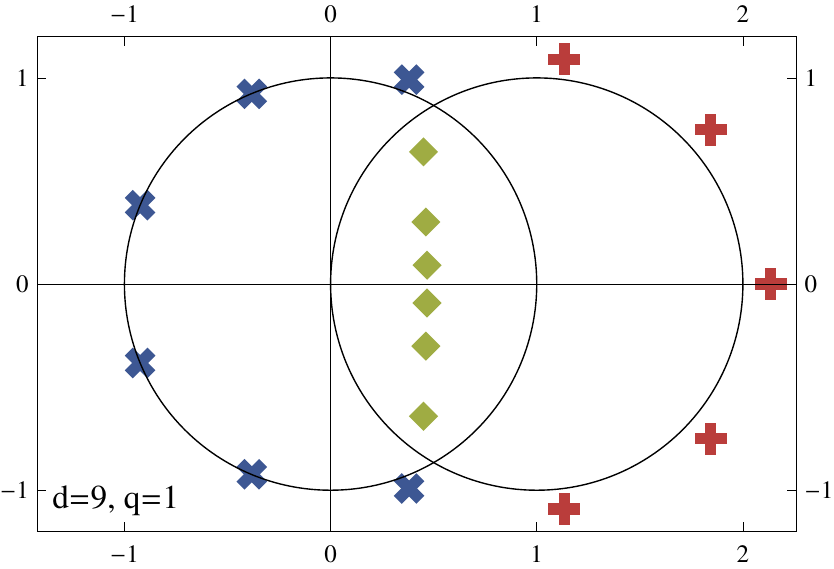} \hfill
\includegraphics[scale=.625]{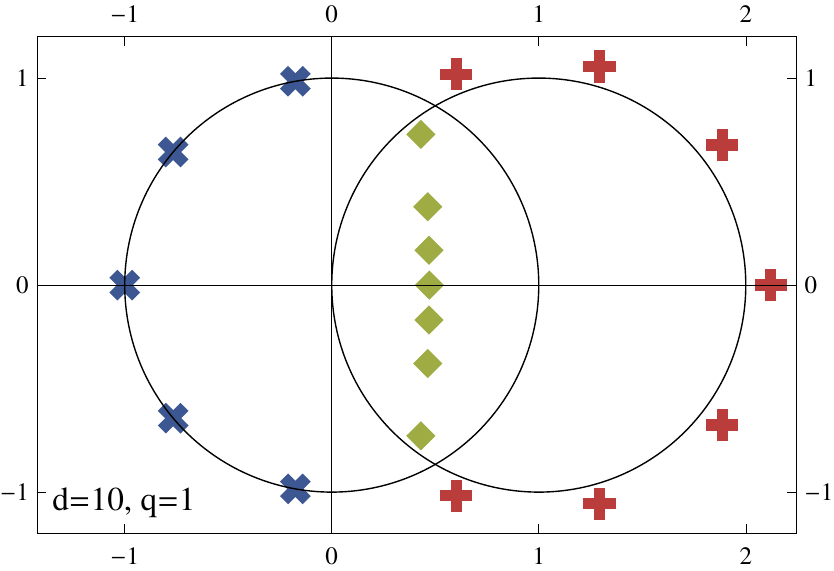}
\includegraphics[scale=.625]{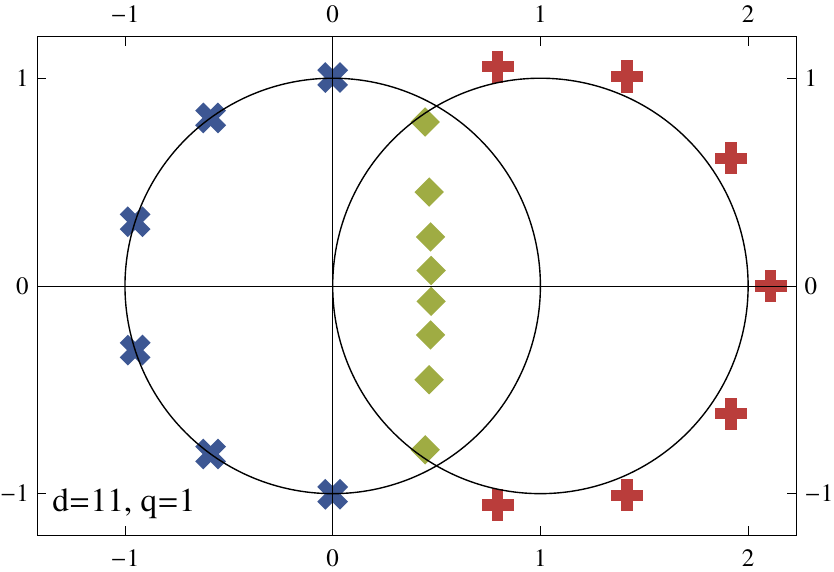} \hfill
\includegraphics[scale=.625]{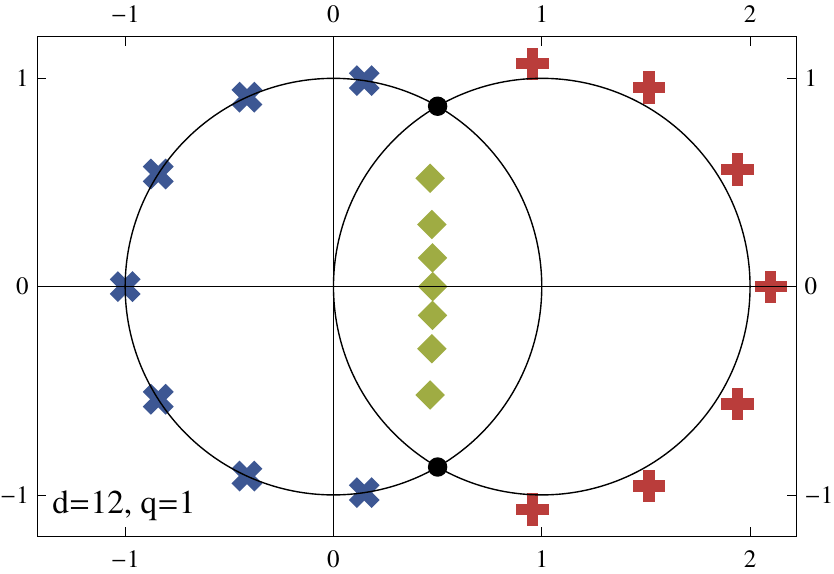}
\end{minipage}
\hfill
\begin{minipage}{0.25\linewidth}
\small
\begin{tabular}{rr|rrrr}
 $d$ & $n$ & $N_1$ & $N_2$ & $N_3$ \\
     &          & {\tiny($\times$)} & {\tiny($\blacklozenge$)} & {\tiny($+$)} \\
\hline
 1 & 1 & 0 & 0 & 1 \\
 2 & 3 & 1 & 1 & 1 \\
 3 & 5 & 2 & 2 & 1 \\
 4 & 7 & 1 & 3 & 3 \\
 5 & 9 & 2 & 4 & 3 \\
\rowcolor[gray]{.8}
 6 & 11 & 3 & 3 & 3 \\
 7 & 13 & 4 & 4 & 5 \\
 8 & 15 & 5 & 5 & 5 \\
 9 & 17 & 6 & 6 & 5 \\
 10 & 19 & 5 & 7 & 7 \\
 11 & 21 & 6 & 8 & 7 \\
\rowcolor[gray]{.8}
 12 & 23 & 7 & 7 & 7 \\
 $\vdots$ & $\vdots$ & $\vdots$ & $\vdots$ & $\vdots$ \\
\rowcolor[gray]{.8}
 42 & 83 & 27 & 27 & 27
\end{tabular}
\end{minipage}
\end{center}
\captionlistentry[table]{. . . }
\captionsetup{labelformat=andtable}
\caption{\label{fig6} Counting the zeros of $G(d; z)$, $n = 2d - 1$. Two additional zeros appear when $6$ divides $d$, see shaded entries in table.}
\end{figure}

For this purpose it is convenient to define $\delta_{6 \mid d} \DEF 1$ if $6$ divides $d$ and $\delta_{6 \mid d} \DEF 0$ otherwise. 

\begin{prop} \label{prop:Zeros}
If $d$ is even, $G(d;z)$ has exactly $4(\lfloor (d-1)/6\rfloor+\delta_{6 \mid d})+1$ zeros on the unit circle $\mathcal{C}_0$; all are simple and satisfy $\re z \leq 1/2$. Their positions are determined by the solutions of the equation $g(\theta) = f(\theta)$ given in \eqref{ThetaEquation} below.

If $d$ is odd, there are no zeros of $G(d; z)$ on the unit circle $\mathcal{C}_0$.
\end{prop}

\begin{proof}
Since $z=1$ is not a zero of $G(d;z)$, the equation $G(d;z)=0$ is equivalent to
\begin{equation*}
z^{d-1}+1=\frac{z+1}{z-1}\frac{z^{d-1}}{(z-1)^{d-1}}.
\end{equation*}
Substituting $z=e^{i\theta}$ and changing to trigonometric functions we arrive at
\begin{equation} \label{ThetaEquation}
g(\theta) \DEF (-1)^{d/2} \cos\left( \frac{d-1}{2} \theta \right) = (-1)^{d/2} \frac{\cos(\theta/2)}{\left[2 i \sin(\theta/2) \right]^d} \FED f(\theta).
\end{equation}
Suppose $d$ is even. Then $f(\theta) = \cos( \theta / 2 ) / [ 2 \sin( \theta / 2 ) ]^d$ is real-valued. Since the complex zeros of $G(d;z)$ occur in complex conjugate pairs, we may assume $0\leq \theta \leq \pi$. Using elementary calculus one shows that the function $f(\theta)$ is monotone decreasing and convex on $(0,\pi)$. Hence, $f(\theta)\geq f(\pi/3)+f'(\pi/3)(\theta-\pi/3)$ on $(0,\pi/3)$. This implies that $f(\theta)>1$ for $\theta\in (0, \pi/3-\alpha)$, where $\alpha=(4-2\sqrt{3})/(3d+1)$. On the other hand, on $[ \pi/3-\alpha,\pi/3)$  one can show that $|g(\theta)|<f(\theta)$. Hence, in \eqref{ThetaEquation} we have only to consider the range $[\pi/3,\pi]$ whereupon $f(\theta) \leq f(\pi/3) = \sqrt{3}/2 < 1$. 

The function $g$ is a cosine function with period $4\pi/(d-1)$, where the sign factor ensures that $g$ has a positive derivative at $\theta = \pi$ (which is also an intersection point of $f$ and $g$). The other intersection points of $f$ and $g$ occur in the half-periods where $g \geq 0$. On such a half-period $I$ the convex function $f$ and the concave function $g$ (when restricted to $I$) can intersect in at most two points counting multiplicity, as can be seen by applying Rolle's theorem to the strictly convex function $f-g$. 
That there are at least two such points on $I$ can be seen from the intermediate value theorem applied to the same function ($f-g>0$ at the endpoints of $I$ and $f-g<0$ at the midpoint of $I$) as illustrated in Figure~\ref{fig:typical.cases} for the three canonical cases.
\begin{figure}[h!t]
\begin{center}
\includegraphics[scale=.55]{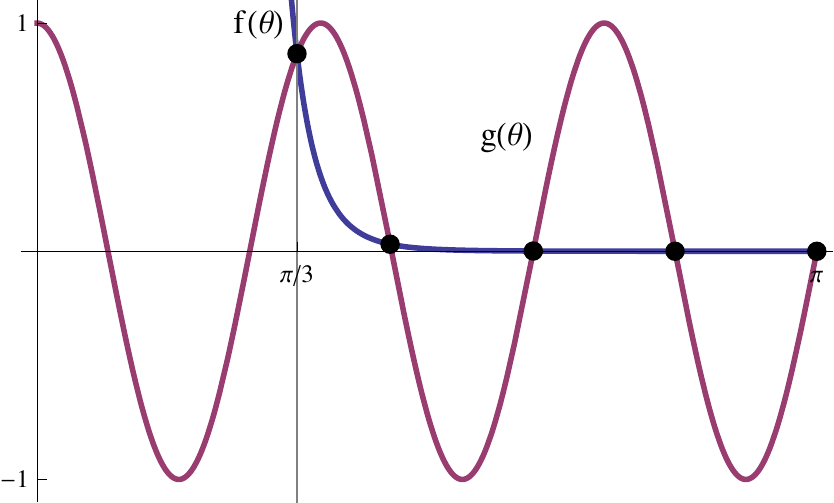}
\includegraphics[scale=.55]{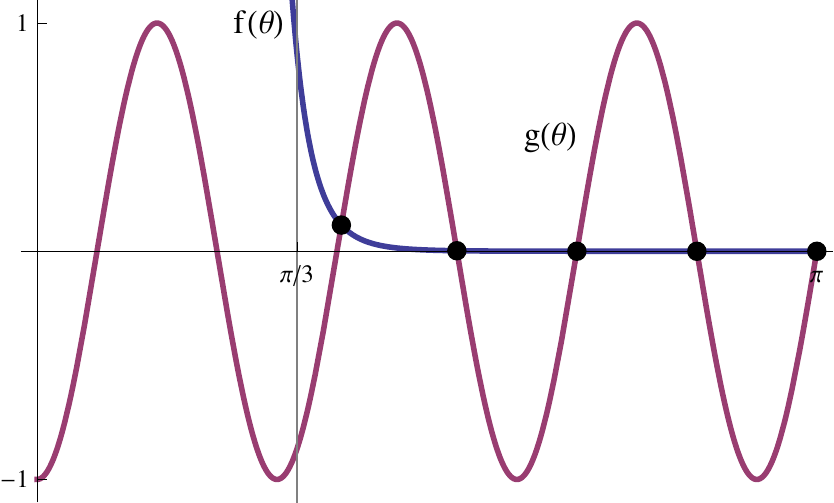}
\includegraphics[scale=.55]{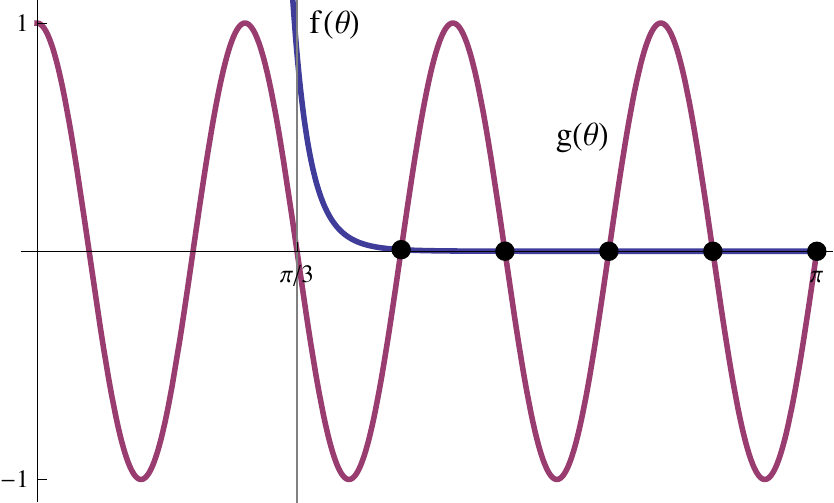}
\end{center}
\caption{\label{fig:typical.cases} Graphs of $f(\theta)$ and $g(\theta)$; typical cases when solving equation~\eqref{ThetaEquation}.} 
\end{figure}

Since there are $\lfloor (d-1)/6\rfloor$ full periods of $g$ in $[\pi/3,\pi]$ plus one more period (partially contained in $[\pi/3,\pi]$) whenever $6$ divides $d$, we have $2(\lfloor (d-1)/6\rfloor+\delta_{6 \mid d})$ zeros in the upper half-plane and that many conjugate zeros of $G(d;z)$ in the lower half-plane.

For odd $d$, equation~\eqref{ThetaEquation} has no real solution in $(0,\pi]$. Since $1$ is not a zero of $G(d; z)$, there are no zeros of $G(d; z)$ on $\mathcal{C}_0$.
\end{proof}

\begin{thm}
All zeros of $G(d; z)$ are simple for each $d \geq 1$.
\end{thm}

\begin{proof}
The cases $d=1$ and $d=2$ are obvious, so assume $d\geq 3$. Let
$\zeta$ be a zero of $G(d; z)$. Then $\zeta$ is simple if
$G^\prime(d; \zeta) \neq 0$. By means of some helpful substitutions for the expressions in braces\footnote{The replacements are $(\zeta-1)^d(\zeta^{d-1}+1) \mapsto \zeta^d + \zeta^{d-1}$ , $[(\zeta-1)^d - \zeta - 1] \zeta^{d-1} \mapsto -(\zeta-1)^{d}$ and $\zeta^d - ( \zeta - 1 )^d \mapsto \zeta^{d-1} [ ( \zeta - 1 )^d - 1]$.}, we find that
\begin{align*}
\zeta \left( \zeta - 1 \right) G^\prime(d; \zeta)
&= d \zeta \left\{ \left( \zeta - 1 \right)^d \left( \zeta^{d-1}
+ 1 \right) \right\} - \zeta^d \left( \zeta - 1 \right) +
\left( d - 1 \right) \left\{ \left[ \left( \zeta - 1 \right)^d -
\zeta - 1 \right] \zeta^{d-1} \right\} \left( \zeta - 1 \right) \notag \\
&= \left( d - 1 \right) \left[ \zeta^{d+1} - \left( \zeta -
1 \right)^{d+1} \right] + \left( d + 1 \right) \zeta^d \\
&= \left( d - 1 \right) \zeta \left\{ \zeta^d - \left( \zeta - 1 \right)^d \right\} + \left( d - 1 \right) \left( \zeta - 1 \right)^d + \left( d + 1 \right) \zeta^d \\
&= \left( d - 1 \right) \left( \zeta - 1 \right)^d \left( \zeta^d + 1 \right) + 2 \zeta^d \FED P(d; \zeta). %\label{prop:simplicity0}
\end{align*}
Suppose to the contrary that $P(d;\zeta) = 0$. Then on replacing $(\zeta-1)^d$ by $-2 \zeta^d /[(d-1)(\zeta^d+1)]$ in the formula \eqref{polynomial.G} for $G(d;\zeta)$ we get
\begin{equation}  \label{eq:G.EQ.Q}
- \left( \zeta^d + 1 \right) G(d; \zeta) = \zeta^{d-1} \left( \zeta^{d+1} + \frac{d+1}{d-1} \zeta^d + \frac{d+1}{d-1} \zeta +1 \right).% =: w^{d-1} Q(d; w).
\end{equation}
The polynomial $Q(d; z)$ of degree $d+1$ obtained by replacing $\zeta$ by $z$ in the second parenthetical expression in \eqref{eq:G.EQ.Q} has three real zeros (at $-1$) and $d-2$ complex zeros for even $d$ and two negative zeros ('near' $-1$) and $d-1$ complex zeros for odd $d$.\footnote{Interestingly, if $\frac{d+1}{d-1}$ is changed to $\frac{d-3}{d-1}$ in $Q(d; z)$, then the new polynomial has all its zeros on the unit circle.} This can be seen from the facts that
\begin{equation*}
Q(d; -1) = - \frac{2\left[ 1 - (-1)^d \right]}{d-1}, \qquad Q^\prime(d; -1) = \left[ 1 - (-1)^d \right] \frac{d+1}{d-1}, \qquad Q^{\prime\prime}(d; -1) = 0
\end{equation*}
and Descartes' Rule of Signs. Substituting $w = e^{i \phi}$ and using trigonometric functions, we arrive at
\begin{equation*}
Q(d; e^{i \phi}) = \frac{4 e^{i [(d+1)/2] \phi}}{d-1} \cos \frac{\phi}{2} \cos \frac{d \phi}{2} \left( d + \tan \frac{\phi}{2} \tan \frac{d \phi}{2} \right).
\end{equation*}
By symmetry, the number of solutions in the open set $(0,2\pi) \setminus \{ \pi \}$ of the equation
\begin{equation} \label{TheotherEquation}
\tan \frac{\phi}{2} \tan \frac{d \phi}{2} = - d
\end{equation}
equals the number of zeros of $\tan(d \phi/2)$ in $(0, 2 \pi)$ (cf. Figure~\ref{fig:simplicity}), which is $d-2$ for even $d$ and $d-1$ for odd $d$.
\begin{figure}[h!t]
\begin{center} 
\includegraphics[scale=.75]{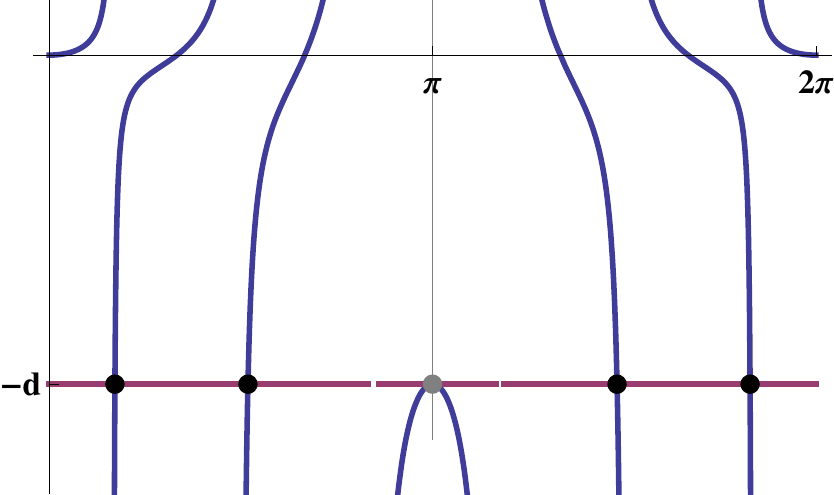} \hspace{10mm}
\includegraphics[scale=.75]{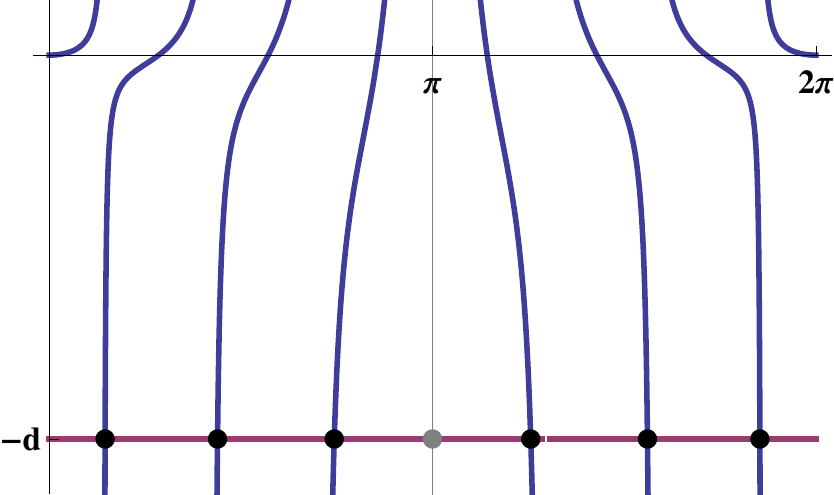} 
\end{center}
\caption{\label{fig:simplicity} Typical cases when solving equation~\eqref{TheotherEquation}.}
\end{figure}
Thus, all zeros of $Q(d; z)$ are accounted for. They are either negative or complex conjugate pairs of zeros located on the unit circle. 

For odd $d$, no zero $\zeta$ of $G(d; z)$ is on the unit circle (Proposition~\ref{prop:Zeros}) or negative (Proposition~\ref{prop:real.zeros}), so $Q(d; \zeta) \neq 0$, which contradicts \eqref{eq:G.EQ.Q}. Suppose $d$ is even. If $\zeta$ is on the unit circle $\mathcal{C}_0$, then it is simple (Proposition~\ref{prop:Zeros}). If $\zeta$ is not on $\mathcal{C}_0$, then $Q(d; \zeta) \neq 0$, which again contradicts \eqref{eq:G.EQ.Q}. 
\end{proof}

\subsection{Asymptotics of Gonchar polynomials}
\label{sec:asymptotics}

Numerically computing the zeros for $G(d;z)$ for small values of $d$, we observe (cf. Figures~\ref{fig6} and \ref{fig:limitdistr} and Table~1) that they essentially form three groups separated by the sets
\begin{subequations} 
\begin{align}
A_1 &\DEF \left\{ z \in \mathbb{C} : \re z < 1/2, | z - 1 | > 1 \right\}, \\
A_2 &\DEF \left\{ z \in \mathbb{C} : | z | < 1, | z - 1 | < 1 \right\}, \\
A_3 &\DEF \left\{ z \in \mathbb{C} : \re z > 1/2, | z | > 1
\right\}.
\end{align}
\end{subequations}
For the purpose of asymptotic analysis (large $d$) we rewrite the
equation $G(d; z) = 0$ in three different ways to emphasize an
exponentially decaying right-hand side when considering zeros 
of $G(d;z)$ from the indicated part of the complex plane:
\begin{subequations} \label{eq:asympt.equations}
\begin{align}
z^{d-1} + 1 &= \frac{z+1}{z-1} \left( \frac{z}{z-1} \right)^{d-1},
& ( \re z &< 1 / 2 ) \label{eq:approx1} \\
\left( z - 1 \right)^d - \left( z + 1 \right) z^{d-1} &= -
\left( z - 1 \right) \left[ \left( z - 1 \right) z \right]^{d-1},
& ( \left| z - 1 \right| \left| z \right| &< 1 ) \label{eq:approx2} \\
\left( z - 1 \right)^d - z - 1 &= - \left( z - 1 \right) \left( \frac{z - 1}{z} \right)^{d-1}. & ( \re z &>
1 / 2 ) \label{eq:approx3}
\end{align}
\end{subequations}

The following theorem concerning the limit behavior of the zeros of $G(d;z)$ as $d \to \infty$ is illustrated in Figure~\ref{fig:limitdistr}. 
\begin{thm}
Let $\Gamma$ be the set consisting of the boundary of the union of the two unit disks centered at $0$ and $1$ and the line-segment connecting the intersection points as indicated in Figure~\ref{fig:limitdistr}. Then, as $d \to \infty$, all the zeros of $G(d;z)$ tend to $\Gamma$, and every point on $\Gamma$ attracts zeros of these polynomials.
\end{thm}

\begin{figure}[h!t]
\centerline{\includegraphics[scale=1]{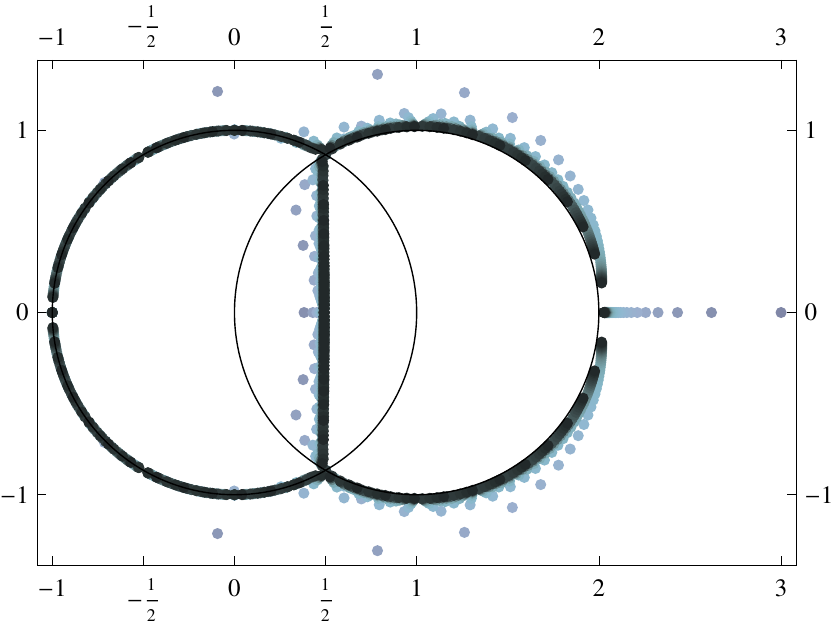}}
\caption{\label{fig:limitdistr} Zeros of $G(d;z)$ for $d = 1, 2, \dots, 40$.}
\end{figure}

\begin{proof}
First, we observe that a closed set $K$ in $\mathbb{C} \setminus \Gamma$ is free of zeros of $G(d;z)$ for sufficiently large $d$ as can be seen from the following relations obtained from \eqref{eq:approx1}:
\begin{align*}
\lim_{d \to \infty} \left| 1 + \frac{1}{z^{d-1}} - \frac{z+1}{z-1} \left( \frac{1}{z-1} \right)^{d-1} \right|^{1/d} &= 
\begin{cases}
1 & \text{if $| z | > 1$, $| z - 1 | > 1$,} \\
1 / \left| z - 1 \right| & \text{if $| z | > 1$, $| z - 1 | < 1$,}
\end{cases} \\
\lim_{d \to \infty} \left| \frac{z+1}{z-1} - \left( \frac{z - 1}{z} \right)^{d-1} \left( z^{d-1} + 1 \right) \right|^{1/d} &=
\begin{cases}
1 & \text{if $| z | < 1$ , $| \frac{z-1}{z} | < 1$ ($\re z > 1/2$),} \\
\left| \frac{z-1}{z} \right| & \text{if $| z | < 1$ , $| \frac{z-1}{z} | > 1$ ($\re z < 1/2$),}
\end{cases} \\
\lim_{d \to \infty} \left| z^{d-1} + 1 - \frac{z+1}{z-1} \left( \frac{z}{z-1} \right)^{d-1} \right|^{1/d} &= 
\begin{cases}
\left| z \right| & \text{if $| z | > 1$, $| \frac{z}{z-1} | < 1$ ($\re z < 1/2$),} \\
1 & \text{if $| z | < 1$, $| \frac{z}{z-1} | < 1$ ($\re z < 1/2$).}
\end{cases} 
\end{align*}

The second part of the assertion, that every point of $\Gamma$ attracts zeros, is proved by contradiction. 
Given a supposedly non-attracting point $w$ on $\Gamma$, there is a sufficiently small\footnote{The disk is small enough that its intersection with $\Gamma$ is contained either in $\mathcal{C}_0$, $\mathcal{C}_1$ or the open line-segment connecting the intersection points of $\mathcal{C}_0$ and $\mathcal{C}_1$.} open disk $D_w$ centered at $w$ containing no zeros of $G(d;z)$ for all sufficiently large $d$. It is possible to then define a single-valued analytic branch of the $d$-th root of any of the rational functions whose moduli appear on the left-hand sides above. Thereby, we obtain sequences of functions which are analytic and uniformly bounded in $D_w$. Such sequences form normal families in $D_w$. According to the right-hand sides above, at least one limit function of these families (which is necessarily analytic in $D_w$) will have the property that its modulus is $1$ in one part and is non-constant in the other part of $D_w$ which is separated by $\Gamma$. This gives the desired contradiction, since an analytic function in a domain that has constant modulus on a subdomain must be constant throughout the whole domain. Consequently, each point of $\Gamma$ attracts zeros of $G(d;z)$ as $d \to \infty$.
\end{proof}

It is inviting to compare the zeros of $G(d;z)$ with the ones of the polynomials given at the left-hand sides of \eqref{eq:approx1}, \eqref{eq:approx2} and \eqref{eq:approx3}. Such comparisons will likely lead to a finer analysis of the properties of the zeros of $G(d;z)$.

We conclude this note with some challenging conjectures.

\begin{conj} \label{conj:unit.circle}
For every positive integer $d$, the zeros of $G(d; z)$ form three
groups separated by the sets $A_1$, $A_2$ and $A_3$ except when $6$
divides $d$ in which case one also has the zeros $( 1 \pm \ii
\sqrt{3} ) / 2$. % as shown in Proposition~\ref{prop:cyclot.divisor}.
% Moreover, the zero counting scheme in the Table~1 holds in general.
\end{conj}

When counting the zeros in the sets $A_1$, $A_2$ and $A_3$ a very regular pattern emerges, which can be seen from Table~1. In fact, inspection of this table shows that the values of column $N_1$ (increased by $2$ when $6$ divides $d$) agree with the number of zeros on $\mathcal{C}_0$ obtained in Proposition~\ref{prop:Zeros}. Assuming that this is true for all $d \geq 1$, by self-reprocity of $G(d;z)$, it would follow that $N_1 = N_2 = N_3 = 4k - 1$ if $d = 6 k$. {\em We expect that the zero counting scheme indicated in the Table~1 generalizes to all $d \geq 1$.}

Numerically, the zeros in $A_1$ and $A_3$ can be found near the respective unit circle $\mathcal{C}_0$ and $\mathcal{C}_1$. 
\begin{conj}
If $d$ is even, all the zeros of $G(d;z)$ in $A_1$ are the zeros on the unit circle $\mathcal{C}_0$ given in Proposition~\ref{prop:Zeros}. If $d$ is odd, the zeros of $G(d;z)$ in $A_1$ alternately lie inside and outside $\mathcal{C}_0$. Furthermore, the zeros of $G(d;z)$ in $A_3$ are always outside of $\mathcal{C}_1$.
\end{conj}

\begin{conj} \label{conj:convex.left}
The zeros of $G(d; z)$ in $A_2$ are located on a curve which is convex from the left.
\end{conj}

{\bf Acknowledgment.} The first author is grateful to Don Zagier and
Wadim Zudilin for inspiring this work when attending the workshop
``Geometry and Arithmetic around Hypergeometric Functions'' from
September 28th --- October 4th, 2008, which was made possible by the
Mathematisches Forschungsinstitut Oberwolfach (MFO) and the
Oberwolfach-Leibniz-Fellow Programme (OWFL).

\appendix

\section{}
In the proof of Proposition \ref{prop:uniqueness} we utilize the following.

\begin{lem} \label{PresLemma} Let $d\geq 2$ and $\PT{z}_1\not = \PT{z}_2$ be  two fixed points
in $\mathbb{R}^{d+1}$. Then
\begin{equation} \label{convolution}
\frac{C_d}{|\PT{z}_1-\PT{z}_2|^{d-1}}=\int_{\mathbb{R}^{d+1}}
\frac{1}{|\PT{t}-\PT{z}_1|^d|\PT{t}-\PT{z}_2|^d}\, d \PT{t}=: J
(\PT{z}_1,\PT{z}_2), \qquad C_d = \frac{\pi^{(d+3)/2} \gammafcn((d-1)/2)}{[ \gammafcn(d/2) ]^2}.
\end{equation}
\end{lem}

\begin{proof}
Observe that $ J (\PT{z}_1,\PT{z}_2)$ is a convergent integral for
any $\PT{z}_1\not = \PT{z}_2$. First a translation
$\PT{t}=\PT{u}+\PT{z}_2$ and then change to spherical coordinates
yields
\begin{equation*}
J(\PT{z}_1,\PT{z}_2) = \int_{\mathbb{R}^{d+1}}
\frac{1}{|\PT{u}-\PT{z}_1+\PT{z}_2|^d} \, \frac{d
\PT{u}}{|\PT{u}|^d} = \omega_d \int_0^\infty \left\{
\int_{\mathbb{S}^d} \frac{d \sigma_d( \PT{\bar{u}} )}{\left( \rho^2
- 2 \rho r \, \PT{\bar{u}} \cdot \PT{\bar{x}} + r^2 \right)^{d/2}}
\right\} \, d \rho,
\end{equation*}
where we used the notation $\rho:=|\PT{u}|$,
$r:=|\PT{z}_1-\PT{z}_2|$, $\PT{\bar{u}}:=\PT{u}/\rho$,
$\PT{\bar{x}}:=(\PT{z}_1-\PT{z}_2)/r$ and
\begin{equation*}
\omega_d = 2 \pi^{(d+1)/2} \big/ \gammafcn((d+1)/2)
\end{equation*}
is the surface area of the unit sphere $\mathbb{S}^d$. Using the
Funk-Hecke formula \cite[p.~20]{Mu1966}, we derive
\begin{equation*}
J(\PT{z}_1,\PT{z}_2) = \omega_{d-1} \int_0^\infty \left\{
\int_{-1}^1 \frac{\left( 1 - v^2 \right)^{d/2-1} d v}{\left( \rho^2
- 2 \rho r \, v + r^2 \right)^{d/2}} \right\} d \rho = \omega_{d-1}
\int_{-1}^1 H(v) \left( 1 - v^2 \right)^{d/2-1} d v,
\end{equation*}
where the formula \cite[Eq.~2.2.9.7]{PrBrMa1986I} enables us to
compute
\begin{equation*}
H(v) := \int_0^\infty \frac{d \rho}{\left( \rho^2 - 2 \rho r \,
 v + r^2 \right)^{d/2}} = r^{1-d} \frac{1}{d-1} \Hypergeom{2}{1}{1/2,(d-1)/2}{(d+1)/2}{1-v^2}.
\end{equation*}
Symmetry and a change of variable $u = 1 - v^2$ leads to (recall
$r=|\PT{z}_1-\PT{z}_2|$)
\begin{align*}
C_d &= |\PT{z}_1-\PT{z}_2|^{d-1} \, J(\PT{z}_1,\PT{z}_2) =
\frac{\omega_{d-1}}{d-1} \int_0^1 u^{d/2-1} \left( 1 - u
\right)^{1/2-1} \Hypergeom{2}{1}{1/2,(d-1)/2}{(d+1)/2}{u} d u,
\intertext{where the integral can be expressed as a
$3F2$-hypergeometric function at unity (cf.
\cite[Eq.~2.21.1.5]{PrBrMa1990III})} &= \frac{\omega_{d-1}}{d-1} \,
\betafcn(d/2,1/2) \,
\Hypergeom{3}{2}{1/2,(d-1)/2,d/2}{(d+1)/2,(d+1)/2}{1} =
\frac{\omega_{d}}{d-1} \, \pi \left[
\frac{\gammafcn((d+1)/2)}{\gammafcn(d/2)} \right]^2.
\end{align*}
The last step follows from the relation $\betafcn(d/2,1/2) =
\omega_{d} / \omega_{d-1}$ and the generalized hypergeometric
function can be evaluated using \cite[Eq.~7.4.4.19]{PrBrMa1990III}.
The result follows.
\end{proof}

\begin{proof}[Proof of Proposition \ref{prop:uniqueness}]
Suppose $\eta_1$ and $\eta_2$ are two signed equilibria on
$\mathbb{S}^d$ associated with the same external field $Q$. Then
\begin{equation*}
V^{\eta_1}(\PT{x})+Q(\PT{x})= F_1 , \quad
V^{\eta_2}(\PT{x})+Q(\PT{x})= F_2 \qquad \text{for all $\PT{x} \in
\mathbb{S}^d$.}
\end{equation*}
Subtracting the two equations and integrating with respect to
$\eta:=\eta_1-\eta_2$ we obtain
\begin{equation*}
\mathcal{I} (\eta) = \int V^\eta (\PT{x}) \, d \eta
(\PT{x})= \int \int \frac{1}{|\PT{x}-\PT{y}|^{d-1}}\,
d\eta(\PT{x})\, d\eta(\PT{y})= 0.
\end{equation*}
Applying \eqref{convolution} from Lemma \ref{PresLemma} we obtain
\begin{align}
0 &= \int \int \frac{1}{|\PT{x}-\PT{y}|^{d-1}}\, d\eta(\PT{x})\,
d\eta(\PT{y})\nonumber \\
&= \frac{1}{C_d}\int \int \left( \int_{\mathbb{R}^{d+1}}
\frac{1}{|\PT{x}-\PT{z}|^d}\frac{1}{|\PT{z-y}|^d} \,
d\PT{z} \right) \,d\eta(\PT{x})\, d\eta(\PT{y}) \nonumber \\
&=\frac{1}{C_d}\int_{\mathbb{R}^{d+1}} \left[ \int
\frac{1}{|\PT{x}-\PT{z}|^d} \, d\eta(\PT{x}) \right]^2 \,
d\PT{z},\label{EnergyZero}
\end{align}
where the interchange of the integration is justified because signed
equilibria have a.e. finite potentials. From \eqref{EnergyZero} we 
conclude for the Riesz potential of $\eta$
\begin{equation*}
\int \frac{1}{|\PT{x}-\PT{z}|^d} \, d\eta(\PT{x}) =0 \quad {\rm a.e.
\ in} \ \mathbb{R}^{d+1}.
\end{equation*}
It turns out that linear combinations of $\{
|\PT{x}-\PT{z}|^{-d}\}_{z\in\mathbb{R}^{d+1}}$ are a dense class in
the space of continuous functions on $\mathbb{S}^d$ (see \cite[p.
214]{La1972}), we obtain that $\eta \equiv 0$ (see also
\cite[Theorem~1.12]{La1972}).
\end{proof}

\begin{proof}[Proof of Asymptotic \eqref{AsymptoticEquation}]
Let $\zeta_d$ denote the largest real zero of $G(d,q;z)$ (that is $R_q = \zeta_d$). We may assume (see discussion regarding appearance of the number $2$) that 
\begin{equation*}
\zeta_d = 2 + f_d \qquad \text{with $f_d = f(\zeta_d) \to 0$ as $d \to \infty$.}
\end{equation*}
Rewriting the equation $G(d,q;\zeta_d) = 0$ in the following way using above relation and exploiting the exponential decay of the right-hand side yields
\begin{equation*}
\left( 1 + f_d \right)^d - q \left( 3 + f_d \right) = - q \left( 1 + f_d \right) \left( \frac{1 + f_d}{2 + f_d} \right)^{d-1} = q \, \mathcal{O}((2/3)^d) \qquad \text{as $d \to \infty$,}
\end{equation*}
or equivalently,
\begin{equation*}
\log( 1 + f_d ) = \frac{1}{d} \log( 3 q ) + \frac{1}{d} \log ( 1 + \frac{f_d}{3} + \mathcal{O}((2/3)^d) ) \qquad \text{as $d \to \infty$.}
\end{equation*}
Now we can use the series expansion of the logarithm function to get a relation for $f_d$:
\begin{equation*}
f_d + \mathcal{O}(f_d^2) = \frac{1}{d} \log( 3 q ) + \frac{1}{d} \mathcal{O}(f_d) \qquad \text{as $d \to \infty$.}
\end{equation*}
The asymptotic form \eqref{AsymptoticEquation} follows.
\end{proof}

\def\cprime{$'$} \def\polhk#1{\setbox0=\hbox{#1}{\ooalign{\hidewidth
  \lower1.5ex\hbox{`}\hidewidth\crcr\unhbox0}}}

\end{document}